\documentclass[final, nomarks]{dmtcs-episciences}

\usepackage[utf8]{inputenc}
\usepackage{subfigure}

\usepackage{amssymb}
\usepackage{amsfonts}
\usepackage{amsmath}
\usepackage{amsthm}
\usepackage{color}
\definecolor{orange}{rgb}{1.0, 0.5, 0.0}
\usepackage{hyperref}
\usepackage{mathtools}

\usepackage[round]{natbib}

\newcommand{\word}{\Sigma^*}
\newcommand{\wordset}{\mathcal{W}}
\newcommand{\roots}{r}

\newcommand{\process}{\mathsf{Process}}

\newcommand{\intv}[1]{\left [ #1 \right ]}
\newcommand{\tw}{{\sf{tw}}}
\newcommand{\w}{{\sf{w}}}

\newcommand{\graphs}{\mathcal{G}}

\newcommand{\graphproblem}{\mathcal{H}}
\newcommand{\acore}{\mathfrak{C}}
\newcommand{\agraph}{G}
\newcommand{\accepting}{\mathsf{Accept}}
\newcommand{\timecomplexity}{\tau}
\newcommand{\size}{\mathsf{Size}}

\newcommand{\partition}{\mathsf{VertPart}}
\newcommand{\edgepartition}{\mathsf{EdgePart}}
\newcommand{\graphpartition}{\mathsf{GraphPart}}
\newcommand{\Ocal}{\mathcal{O}}

\newtheorem{theorem}{Theorem}
\newtheorem{definition}[theorem]{Definition}
\newtheorem{lemma}[theorem]{Lemma}
\newtheorem{corollary}[theorem]{Corollary}
\newtheorem{observation}[theorem]{Observation}

\author{Julien Baste}
\title{Composing dynamic programming~ \\ tree-decomposition-based algorithms}
\affiliation{
  Univ. Lille, CNRS, Centrale Lille, UMR 9189 CRIStAL, F-59000 Lille, France}

\begin{document}
\publicationdata{vol. 26:2}{2024}{1}{10.46298/dmtcs.11069}{2023-03-14;
  2023-03-14; 2024-01-25}{2024-02-06}

\maketitle

\begin{abstract}

\medskip

  Given two integers $\ell$ and $p$ as well as $\ell$ graph classes $\mathcal{H}_1,\ldots,\mathcal{H}_\ell$, the problems
$\mathsf{GraphPart}(\mathcal{H}_1, \ldots, \mathcal{H}_\ell,p)$, \break
$\mathsf{VertPart}(\mathcal{H}_1, \ldots, \mathcal{H}_\ell)$, and
$\mathsf{EdgePart}(\mathcal{H}_1, \ldots, \mathcal{H}_\ell)$ ask, given graph $G$ as input, whether $V(G)$, $V(G)$, $E(G)$ respectively can be partitioned into $\ell$ sets $S_1, \ldots, S_\ell$ such that, for each $i$ between $1$ and $\ell$, $G[S_i] \in \mathcal{H}_i$, $G[S_i] \in \mathcal{H}_i$,  $(V(G),S_i) \in \mathcal{H}_i$ respectively.
Moreover in $\mathsf{GraphPart}(\mathcal{H}_1, \ldots, \mathcal{H}_\ell,p)$, we request that the number of edges with endpoints in different sets of the partition is bounded by $p$.
We show that if there exist dynamic programming tree-decomposition-based algorithms for recognizing the graph classes $\mathcal{H}_i$, for each $i$, then we can constructively create a dynamic programming tree-decomposition-based algorithms for
$\mathsf{GraphPart}(\mathcal{H}_1, \ldots, \mathcal{H}_\ell,p)$,
$\mathsf{VertPart}(\mathcal{H}_1, \ldots, \mathcal{H}_\ell)$, and
$\mathsf{EdgePart}(\mathcal{H}_1, \ldots, \mathcal{H}_\ell)$.
We apply this approach to known problems.
For well-studied problems, like \textsc{Vertex Cover} and \textsc{Graph $q$-Coloring}, we obtain running times that are comparable to those of the best known problem-specific algorithms.
For an exotic problem from bioinformatics, called \textsc{DisplayGraph}, this approach improves the known algorithm parameterized by treewidth.
\end{abstract}

\section{Introduction}

In one of the first graph partition problems, one is asked, given a graph $G$ and an integer $k$, whether $V(G)$ can be partitioned into $\ell$ sets $V_1, \ldots, V_\ell$ such that the number of edges between two different sets is small, see for instance~\cite{GoHo1994}.
These problems have many applications starting from clustering genes by~\cite{ShMaSh2003}, through optimizing financial
problems by~\cite{MeWe2012}, parallel scientific computing to image segmentation by~\cite{GrSc2006} and \cite{ToMo2012}, and analysis of social networks by~\cite{QiYaWa2010}.
The above specified graph partitioning problem favors cutting small sets of isolated vertices in the input graph as shown by~\cite{ShMa2000} and \cite{WuLe1993}.
In order to avoid this kind of solution which is often undesirable for many practical applications, restrictions are often imposed on the sets $V_i$, $i \in \intv{1,\ell}$.
The most natural restriction is to require the partition to be balanced as done by~\cite{AnRa2006}.
Another one, used in image segmentation, is to consider normalized cuts as done by~\cite{ShMa2000}, that is, cuts that maximize the similarity within the sets while minimizing the dissimilarity between the sets.
In social networks, the graph clustering problem is a graph partition problem where the graphs $G[V_i]$, $i \in \intv{1,\ell}$, are required to be dense as studied by~\cite{Sc2007}.
In this paper, we consider the graph partition problem in a general form defined in the following way.
Given $\ell$ graph classes $\mathcal{H}_1, \ldots, \mathcal{H}_\ell$ and an integer $p$, the $\graphpartition(\mathcal{H}_1, \ldots, \mathcal{H}_\ell,p)$ problem consists in, given a graph $G$, determining whether $V(G)$ can be partitioned into $\ell$ sets $V_1,\ldots, V_\ell$ such that $\{\{u,v\} \in E(G) \mid u \in V_i, v \in V_j, i \not = j\}$, \textit{i.e.}, the set of \emph{transversal} edges, is of size at most $p$ and $G[V_i] \in \mathcal{H}_i$ for each $i \in \intv{1,\ell}$.

Coloring problems are special kinds of graph partition problems where the number of transversal edges  is not relevant anymore.
So, in the $\partition(\mathcal{H}_1, \ldots, \mathcal{H}_\ell)$ problem,
the task is to determine whether the vertex set of the input graph $G$ can be partitioned into $\ell$ sets $V_1, \ldots, V_\ell$ such that $G[V_i] \in \mathcal{H}_i$ for each $i \in \intv{1,\ell}$.
The most famous coloring problem is the \textsc{Graph $3$-Colorability} problem
corresponding to $\partition(\mathcal{I},\mathcal{I},\mathcal{I})$ where
$\mathcal{I}$ is the class of edgeless graphs.
This problem is one of the first problems proved to be \textsf{NP}-hard by~\cite{Ka1972} and has attracted a lot of attention.
While \textsc{Graph $3$-Colorability} is the best known, several other graph classes are also under study.
For instance, \cite{YaYu2004}, \cite{RaSh2013}, and \cite{YuWa2003} consider the induced matching partition where each vertex set of the partition should induce a graph of maximum degree $1$.
\cite{ChChCh2004} focus on
 $\partition(\mathcal{H}_1, \ldots, \mathcal{H}_\ell)$ where $\ell$ is a fixed integer, $\mathcal{H}_1= \ldots =  \mathcal{H}_\ell = \mathcal{R}$, and $\mathcal{R}$ is either the class of every tree or the class of every forest.
These problems are called \textsc{Tree Arboricity} when $\mathcal{R}$ is the class of every tree and \textsc{Vertex Arboricity} when $\mathcal{R}$ is the class of every forest.
 They provide polynomial time algorithms for block-cactus graph, series-parallel graphs, and cographs.
\cite{YaYu2007} focus on planar graphs of diameter two.
As shown by \cite{JaJoKeStWu2019}, these problems have, in particular, applications in bioinformatics for constructing phylogenetic trees.

\bigskip
The treewidth of a graph is a structural parameter that measures the similarity of the graph to a forest, see Section~\ref{sec:prelim} for the formal definitions.
\cite{Co1990} shows that every problem that can be expressed in monadic second-order logic can be solved in \textsf{FPT}-time parameterized by treewidth, \textit{i.e.}, in time $f(\tw)\cdot n^{\mathcal{O}(1)}$ for some function $f$ where $n$ (resp. $\tw$) denotes the number of vertices (resp. the treewidth) of the input graph.
\cite{Ra2007} shows that if there is a monadic second-order logic formula that recognizes a graph class $\mathcal{H}$, then for any fixed integer $\ell$, $\partition(\mathcal{H}_1,\ldots, \mathcal{H}_\ell)$, with $\mathcal{H}_i = \mathcal{H}$ for all $i \in \intv{1,\ell}$, can be solved in polynomial time on graphs of bounded treewidth.
If~\cite{Co1990} and~\cite{Ra2007} provide powerful meta-algorithms, the claimed running times may be far from being optimal.
For instance, \cite{Co1990} provides an $2^{2^{\mathcal{O}(\tw)}}\cdot n^{\mathcal{O}(1)}$ algorithm for \textsc{Graph $3$-Coloring} when it is well known that an $2^{{\mathcal{O}(\tw)}}\cdot n^{\mathcal{O}(1)}$ algorithm exists, see for instance~\cite[Theorem 7.9]{CyFoKoLoMaPiPiSa2015}.

Recently, treewidth has found several applications in bioinformatics when dealing with the so-called display graphs~ as illustrated in the work of \cite{BrLa2006}, \cite{ScIeKeBr2014}, and \cite{fvBaPaSaSc2017}.
In order to solve the \textsc{DisplayGraph} problem, \cite{JaJoKeStWu2019} want to determine whether a given graph of bounded treewidth is a positive instance of $\partition(\mathcal{T},\mathcal{T})$ where $\mathcal{T}$ is the class of all trees.
Using Courcelle's theorem, they provide a $2^{2^{\mathcal{O}(\tw)}}\cdot n^{\mathcal{O}(1)}$ algorithm.
Using the approach developed in this paper, we obtain, as Corollary~\ref{coro:tt}, an algorithm running in time $2^{\mathcal{O}(\tw)}\cdot n^{\mathcal{O}(1)}$ for this same problem.

\bigskip

The dynamic programming core model is a formalism introduced by~\cite{fvBaFeJaMaOlPhRo2022}.
It was first introduced in order to construct meta-algorithms for what are called diverse problems.
It provides a formalism for dynamic programming algorithms that process a tree decomposition, once, in a bottom-up approach.
This kind  of algorithm is indeed widely used when working with treewidth.
Therefore the dynamic programming core model allows us to manipulate most of the known algorithms that process a tree decomposition.

In the present paper, we use the expressive power of this formalism and show that, with some enhancement, it can be used to easily provide algorithms, with good running times, that solve the graph partition problems parameterized by treewidth.
Roughly speaking, given $\ell$ graph classes $\mathcal{H}_1,\ldots, \mathcal{H}_\ell$, we show that solving $\graphpartition(\mathcal{H}_1,\ldots, \mathcal{H}_\ell,p)$ or $\partition(\mathcal{H}_1,\ldots, \mathcal{H}_\ell)$ is not much harder than recognizing each $\mathcal{H}_i$, using a dynamic programming tree-decomposition-based algorithm.
Moreover, we provide, in Theorem~\ref{th:graphpartition} and Theorem~\ref{th:partition}, the explicit running time needed for solving $\graphpartition(\mathcal{H}_1,\ldots, \mathcal{H}_\ell,p)$ and $\partition(\mathcal{H}_1,\ldots, \mathcal{H}_\ell)$ respectively as a function of the running time needed for recognizing each $\mathcal{H}_i$.
We provide, in Theorem~\ref{th:edgepartition}, similar result for the case where we want to partition the edge set of the graph, that is,
for the graph problem $\edgepartition(\mathcal{H}_1, \ldots, \mathcal{H}_\ell)$ that, given a graph $G$, consists in determining whether $E(G)$ can be partitioned into $\ell$ sets $S_1, \ldots, S_\ell$ such that $(V(G),S_i) \in \mathcal{H}_i$ for each $i \in \intv{1,\ell}$.

The main feature of our contribution is to present a meta-approach that provides easy-to-build and efficient algorithms for exotic problems.
Moreover the running time obtained for known problems are comparable to the best-known algorithms specifically designed for each given known problem.

\bigskip

In Section~\ref{sec:prelim}, we introduce the notations and useful definitions.
Section~\ref{sec:dpcm} is devoted to the definition of the dynamic programming core model together with some examples of dynamic cores.
The main results are given in Section~\ref{sec:main}.
In Section~\ref{sec:application}, we show how these results can be applied to reproduce known results and to provide unknown algorithms for exotic problems.
We provide a short conclusion in Section~\ref{sec:conclusion}.

\section{Preliminaries}
\label{sec:prelim}
We denote by $\mathbb{N}$ the set of nonnegative integers.
Given two integers $a$ and $b$ we define $\intv{a,b}$ the set of every integer $c$ such that $a \leq c \leq b$.
Let $G$ be a graph.
Let $\ell$ be an integer and $A = (m_1,\ldots, m_\ell)$ be a $\ell$-tuple. For each $i \in \intv{1,\ell}$, we use the notation $A.(i)$ to refer to the $i$-th coordinate of $A$, \textit{i.e.}, in this case to $m_i$.
Note that the coordinates are numbered from $1$ to $\ell$.
Given a set $S$, we denote by $2^S$ the set of every subset of $S$.
Given an alphabet $\Sigma$, we denote by $\mathcal{W}_\Sigma$ the set of every finite words over $\Sigma$.
We denote by $\Gamma$ the set of three special letters denoted  ``('', ``)'', and ``,''.

\bigskip

We use $V(G)$ and $E(G)$ to denote the vertex and edge sets, respectively, of the graph $G$.
Through out this paper, we assume that vertices are represented as elements of $\mathbb{N}$.
Given a set $S \subseteq V(G)$, we denote by $G[S]$ the subgraph of $G$ induced by $S$.
Given a set $S \subseteq E(G)$, we denote by $G[S]$ the graph $(V(G), S)$.
Given  two sets $S_1,S_2 \subseteq V(G)$, we denote by $\mathsf{edge}_G(S_1,S_2)$ the set of every edge of $G$ with one endpoint in $S_1$ and the other endpoint in $S_2$.
We also denote by $\mathcal{G}$ the set of every graph.
We denote by $\mathcal{I}$ the class of edgeless graphs.
Given an integer $p$, we denote by $\graphs_p$ the class of every graph with at most $p$ vertices.
We also denote by $\mathcal{T}$ the set of every tree
and by $\mathcal{F}$ the set of every forest.
Given a tree $T$ rooted at $r$, for each $t \in V(T)$, we denote by $\mathsf{child}(t)$ the set of every child of $t$ in $T$ and by $\mathsf{desc}(t)$  the set of every descendent of $t$ in $T$.

A \emph{rooted tree decomposition} of a graph $G$ 
is a tuple ${\cal D}=(T,\roots,{\cal X})$, where $T$ is a tree rooted at $\roots \in V(T)$
and ${\cal X}=\{X_{t}\mid t\in V(T)\}$ is a collection of subsets of $V(G)$
such that:
\begin{itemize}
\item $\bigcup_{t \in V(T)} X_t = V(G)$,
\item for every edge $\{u,v\} \in E$, there is a $t \in V(T)$ such that $\{u, v\} \subseteq X_t$, and
\item for each $\{x,y,z\} \subseteq V(T)$ such that $z$ lies on the unique path between $x$ and $y$ in $T$,  $X_x \cap X_y \subseteq X_z$.
\end{itemize}

The \emph{width} of a  tree decomposition 
${\cal D}=(T,r, {\cal X})$, denoted by $\w(\mathcal{D})$, is defined as \break $\max_{t \in V(T)} |X_t|-1$.
The \emph{treewidth} of a graph $G$, denoted by $\tw(G)$, is the smallest integer $w$ such that there exists a tree decomposition of $G$ of width at most $w$.
We also define $Y_t = X_t \cup \bigcup_{t' \in \mathsf{child}(t)}X_{t'}$ and
$Z_t = X_t \cup \bigcup_{t' \in \mathsf{desc}(t)}X_{t'}$.

It is well known, see for instance~\cite{Kl1994}, that given a rooted tree decomposition $\mathcal{D} = (T,r,\mathcal{X})$, we can, without loss of generality, assume that $X_r = \emptyset$, that, for each $t \in V(T)$, $t$ has at most $2$ children and that $|Y_t| \leq |X_t|+1$.
In the following we always assume that the rooted tree decompositions have these properties.

Given a graph $G$, a rooted tree decomposition $\mathcal{D} = (T,r,\mathcal{X})$ of $G$, and a set $S \subseteq V(G)$, we define
$\mathcal{D}[S]$ to be $(T,r,\{X_t \cap S \mid t \in V(T)\})$.
Note that $\mathcal{D}[S]$ is a rooted tree decomposition of $G[S]$.

\section{Dynamic programming core model}
\label{sec:dpcm}

In this section we define and use an improvement of the dynamic programming core model introduced by~\cite{fvBaFeJaMaOlPhRo2022}.
The main idea of this model is to formalize what is a dynamic programming algorithm based on a tree decomposition.
This will allow us to manipulate these algorithms in their generic form in order to construct meta-algorithms.

\begin{definition}[Dynamic Core]
  \label{definition:DynamicCore}
  A \emph{dynamic core} $\acore$ over an alphabet $\Sigma$ is a set of four functions:
  \begin{itemize}
  \item $\accepting_{\acore}: \graphs \to 2^{\wordset_{\Sigma}}$,
  \item $\process_{\acore,0}: \graphs \to 2^{\wordset_{\Sigma}}$,
  \item $\process_{\acore,1}: \graphs \times \graphs \to 2^{\wordset_{\Sigma}^2}$, and
  \item $\process_{\acore,2}: \graphs \times \graphs \times \graphs \to 2^{\wordset_{\Sigma}^3}$.
  \end{itemize}

\end{definition}

In the following, we always assume that the associated alphabet is implicitly given when a dynamic core is mentioned and
we denote by $\Sigma_{\acore}$ the alphabet associated to a dynamic core $\acore$.
Given a dynamic core $\acore$, a graph $G$ and a rooted tree decomposition $\mathcal{D}=(T,\roots,{\cal X})$ of $G$,
the \emph{data} of $\acore$ associated to $(G,\mathcal{D})$ are, for each $t \in V(T)$:
  \begin{align*}
    \accepting_{\acore,\agraph,\mathcal{D}} &= \accepting_{\acore}(G[X_r])\\
    \process_{\acore,\agraph,\mathcal{D}} (t) &=
\begin{cases}
  \process_{\acore,0} (G[X_t])&  \text{if $\mathsf{child}(t) = \emptyset$,}\\
 \process_{\acore,1} (G[X_t], G[X_{t'}])& \text{if $\mathsf{child}(t) = \{t'\}$, and}\\
 \process_{\acore,2} (G[X_t],G[X_{t'}], G[X_{t''}])& \text{if $\mathsf{child}(t) = \{t',t''\}.$}
\end{cases}
  \end{align*}
  We would like to highlight that Definition~\ref{definition:DynamicCore} is the main addition of this paper, concerning the definition of the dynamic programming core model, compared to~\cite{fvBaFeJaMaOlPhRo2022}.
  The functions $\process_{\acore,0}$, $\process_{\acore,1}$, and $\process_{\acore,2}$ can be viewed as the rules to update the table of the given dynamic programming algorithm, and so allow to easily and naturally construct a dynamic core from a dynamic programming algorithm that is based on a tree decomposition.
  Note that these rules are given independently of the tree decomposition as it is usual to do for a dynamic programming algorithm.
  The definitions of $\accepting_{\acore,\agraph,\mathcal{D}}$ and $\process_{\acore,\agraph,\mathcal{D}}$, if we consider $\Sigma = \{0,1\}$, are similar to the ones initially defined by~\cite{fvBaFeJaMaOlPhRo2022}.

  Note that, in this work, we assume that $X_r = \emptyset$ and so $\accepting_{\acore,\agraph,\mathcal{D}} = \accepting_\acore((\emptyset,\emptyset))$.
  We still keep the general notation, with $X_r$, to keep the setting as flexible as possible.
  Given a function $f : \mathcal{K} \to \mathcal{J}$ and an input $I \in \mathcal{K}$, we denote by $\tau(f,I)$ the time needed to compute $f(I)$.
  Given a dynamic core $\acore$, a graph $G$ and a rooted tree decomposition $\mathcal{D} = (T,r,\mathcal{X})$ we let:
\begin{align*}
  \tau_n(\acore,\agraph,\mathcal{D},t) &=
\begin{cases}
  \max_{A \subseteq Y_t} \tau(\process_{\acore,0},G[A \cap X_t])& \text{if $\mathsf{child}(t) = \emptyset$,}\\
  \max_{A \subseteq Y_t} \tau(\process_{\acore,1}, (G[A \cap X_t],G[A \cap X_{t'}]))& \text{if $\mathsf{child}(t) = \{t'\}$, and}\\
  \max_{A \subseteq Y_t} \tau(\process_{\acore,2}, (G[A \cap X_t],G[A \cap X_{t'}],G[A \cap X_{t''}]))& \text{if $\mathsf{child}(t) = \{t',t''\}.$}
  \end{cases}
\end{align*}
We also define $\tau_g(\acore,\agraph,\mathcal{D}) = \sum_{t \in V(T)}\tau_n(\acore,\agraph,\mathcal{D},t)$
and $\mathsf{size}(\acore,G,\mathcal{D},t) = \max_{A \subseteq Y_t} |\process_{\acore,G[A],\mathcal{D}[A]}(t)|$.
The $n$ of $\tau_n$ stands for node, and the $g$ of $\tau_g$ stands for global.
Note that, for each $t \in V(T)$,  $\process_{\acore,G[Y_t],\mathcal{D}[Y_t]}(t) = \process_{\acore,G,\mathcal{D}}(t)$.
We say that a dynamic core $\acore$ is \emph{polynomial}, if for each graph $G$, each rooted tree decomposition $\mathcal{D}= (T,r,\mathcal{X})$ of $G$, and each $t \in V(T)$, $\mathsf{size}(\acore,G,\mathcal{D},t)$ and $\tau_n(\acore,\agraph,\mathcal{D},t)$ are polynomial in $|V(G)| + |V(T)|$.

\begin{definition}
\label{definition:Witness}
Let $\acore$ be a dynamic core, $\agraph$ be a graph in $\graphs$, and  
$\mathcal{D} = (T,\roots, \mathcal{X})$ be a rooted tree decomposition of 
$\agraph$. A \emph{$(\acore,\agraph,\mathcal{D})$-witness} is a function $\alpha:V(T)\rightarrow \Sigma_{\acore}^*$
such that $  \alpha(r)\in \accepting_{\acore,\agraph,\mathcal{D}}$
and for each $t\in V(T)$,
\begin{align*}
\begin{drcases}
\text{if $\mathsf{child}(t) = \emptyset$,} & \alpha(t)\\
 \text{if $\mathsf{child}(t) = \{t'\}$,} &  (\alpha(t),\alpha(t'))\\
\text{if $\mathsf{child}(t) = \{t',t''\}$,} & (\alpha(t), \alpha(t'), \alpha(t''))
\end{drcases}
                                           & \in \process_{\acore,G,\mathcal{D}}(t)\\
\end{align*}
\end{definition}

The witness provided in Definition~\ref{definition:Witness} can be seen as a proof of the correctness of the algorithms we can produce using a given dynamic core.

\begin{definition}
  \label{definition:Solvability}
  Let $\mathcal{H}$ be a class of graphs.
  We say that a dynamic core $\acore$ \emph{solves} $\mathcal{H}$ if for each graph 
$\agraph\in \graphs$, and each rooted tree decomposition $\mathcal{D}$ of $\agraph$,
$\agraph\in \graphproblem$ if and only if a $(\acore,\agraph,\mathcal{D})$-witness exists. 
\end{definition}

As explained by~\cite{fvBaFeJaMaOlPhRo2022} and summarized in Theorem~\ref{theorem:DynamicSolvability}, a dynamic core can be seen as an algorithm producer.
Moreover the running time of the produced algorithms is directly connected to the definition of the associated dynamic core.

\begin{theorem}[\cite{fvBaFeJaMaOlPhRo2022}]
\label{theorem:DynamicSolvability}
Let $\graphproblem$ be a class of graphs and $\acore$ be a dynamic core that solves $\graphproblem$. 
Given a graph $\agraph\in \graphs$ and a rooted tree decomposition $\mathcal{D} = (T,r,\mathcal{X})$ of 
$\agraph$, one can decide whether $\agraph\in \graphproblem$ in time
$\mathcal{O}\left(\sum_{t \in V(T)}|\process_{\acore,\agraph,\mathcal{D}}(t)| + \timecomplexity_g(\acore,G,\mathcal{D})\right)$.
\end{theorem}

\subsection{Some examples of dynamic core}

In this section we provide a few examples of dynamic cores.
We start by a trivial dynamic core that allows us to produce an algorithm that recognizes that a graph has no edge.
This dynamic core solves $\mathcal{I}$, the class of graphs with no edges.

\begin{observation}
  \label{obs:indep}
  $\mathcal{I}$ can be solved by a polynomial dynamic core $\acore$.
\end{observation}
\begin{proof}
  We define $\acore$ such that for each $G,G',G'' \in \graphs$,
  \begin{align*}
    \accepting_{\acore}(G) & = \{\top\} \\
    \process_{\acore,0}(G) & = \{\top \mid E(G) = \emptyset\}\\
    \process_{\acore,1}(G,G') &= \{(\top,\top) \mid E(G) = \emptyset\}\\
    \process_{\acore,2}(G,G',G'') &= \{(\top,\top,\top) \mid E(G) = \emptyset\}\\
  \end{align*}

  In this case, $\Sigma_\acore = \{\top\}$ where $\top$ represents the fact that the already explored part does not contain any edge.
  Given $G \in \mathcal{I}$ and $\mathcal{D} = (T,r,\mathcal{X})$ a rooted tree decomposition of $G$, a $(\acore,\agraph,\mathcal{D})$-witness $\alpha$ is such that, for each $t \in V(T)$,
$\alpha(t) = \top$.
  
  For the running time, note that given a graph $G$ and a rooted tree decomposition $\mathcal{D} = (T,r,\mathcal{X})$ of $G$, then for each $t \in V(T)$, $|\process_{\acore,G,\mathcal{D}}(t)| \leq 1$ and $\timecomplexity_n(\acore,G,\mathcal{D},t) = \mathcal{O}(|X_t|)$.
\end{proof}

We also provide a slightly more involved dynamic core that solves $\graphs_p$ for some integer $p$, \textit{i.e.}, the class of graphs with at most $p$ vertices.

\begin{observation}
  \label{obs:boudedgraphs}
  Let $p$ be an integer.
  $\graphs_p$ can be solved by a polynomial dynamic core $\acore$.
\end{observation}

\begin{proof}
  We define $\acore$ such that for each $G,G',G'' \in \graphs$,
  \begin{align*}
    \accepting_{\acore}(G) & = \{q \mid q \in \intv{0,p}\} \\
    \process_{\acore,0}(G) & = \{0\}\\
    \process_{\acore,1}(G,G') &= \{(q,q') \mid q' \geq 0,\ q\leq p, \text{ and } q = q' + |V(G') \setminus V(G)| \}\\
    \process_{\acore,2}(G,G',G'') &= \{(q,q',q'') \mid q', q'' \geq 0,\ q\leq p, \text{ and }q = q' + q''+ |(V(G') \cup V(G'')) \setminus V(G)|\}\\
  \end{align*}
  It is now an easy task to show that $\acore$ solves $\graphs_p$.
  In this case, given a graph $\agraph \in \graphs_p$ and a rooted tree decomposition $\mathcal{D} = (T,r,\mathcal{X})$, a possible $(\acore,\agraph,\mathcal{D})$-witness $\alpha$ is such that, for each $t \in V(T)$,   $\alpha(t) = |Z_t \setminus X_t|$.
  Simply note that we only count the number of vertices in the part that has already been completely explored and forgotten and that $X_r = \emptyset$.

  For the running time, note that given a graph $G$ and a rooted tree decomposition $\mathcal{D} = (T,r,\mathcal{X})$ of $G$, then, for each $t \in V(T)$, $|\process_{\acore,G,\mathcal{D}}(t)| \leq (p+1)^2$ and $\timecomplexity_n(\acore,G,\mathcal{D},t) = \mathcal{O}(|\process_{\acore,G,\mathcal{D}}(t)|)$.
  The lemma follows.
\end{proof}

Observations~\ref{obs:indep} and~\ref{obs:boudedgraphs} show how to construct a dynamic core for trivially solvable problems.
We mostly provided these observations as pedagogical examples.
One can then get the intuition that most of the dynamic programming algorithms parameterized by treewidth can be translated into dynamic cores.
Indeed, such an algorithm creates a dynamic programming table for each node of the tree decomposition.
Theses tables are updated depending of dynamic programming tables of the children of the node taken into consideration.
Transposing how these updates are done into a consistent definition of $\process_{\acore,0}$, $\process_{\acore,1}$, or $\process_{\acore,2}$, depending on the number of children of the node taken in consideration, will then provide a dynamic core for the problem.

The rank-based approach, developed by~\cite{BoCyKrNe2015}, provides, in particular, a deterministic algorithm that solves \textsc{Feedback Vertex Set} in time $2^{\mathcal{O}(\tw)}\cdot n^{\mathcal{O}(1)}$, where $n$ (resp. $\tw$) stands for the size (resp. treewidth) of the input graph.
From this algorithm, one can easily obtain a dynamic core for recognizing if a graph is a tree.
We omit the proof of it as it requires to reintroduce several tools presented by~\cite{BoCyKrNe2015} that are out of the scope of this paper.

  \begin{observation}
  \label{obs:tree}
  The class $\mathcal{T}$ of trees can be solved by a dynamic core $\acore$  such that, for each graph $G$, each rooted tree decomposition $\mathcal{D} = (T,r,\mathcal{X})$ of $G$, and each $t \in V(T)$:
    \begin{itemize}
  \item $|\process_{\acore,\agraph,\mathcal{D}}(t)| = 2^{\mathcal{O}(\w(\mathcal{D}))}\cdot n^{\mathcal{O}(1)}$
  \item $\timecomplexity_g(\acore,G,\mathcal{D}) = 2^{\mathcal{O}(\w(\mathcal{D}))}\cdot |V(T)| \cdot n^{\mathcal{O}(1)} $
  \end{itemize}
\end{observation}

Note in particular that the dynamic core provided in Observation~\ref{obs:tree} is not polynomial.

\subsection{Union and intersection of dynamic core}

In this section we provide some simple combinations of dynamic cores.
More precisely, we show how to take the union and the intersection of two dynamic cores.
This will allow us, in Theorems~\ref{th:graphpartition}, \ref{th:partition}, and \ref{th:edgepartition}, to consider the union (resp. intersection) of recognizable classes without having to prove each time that the considered union (resp. intersection) is recognizable.
Let $\mathcal{H}_1$ and $\mathcal{H}_2$ be two graph classes and let $\acore_1$ (resp. $\acore_2$) be a dynamic core that solves $\mathcal{H}_1$ (resp. $\mathcal{H}_2$).
We would like to stress that,
in order to solve
$\mathcal{H}_1 \cup \mathcal{H}_2$ or $\mathcal{H}_1 \cap \mathcal{H}_2$, the naive procedure, consisting of using $\acore_1$ and then using $\acore_2$, would be more efficient with regard to the running time but will not produce a dynamic core.
As the main theorems of the paper, namely Theorems~\ref{th:graphpartition}, \ref{th:partition}, and \ref{th:edgepartition}, rely on the knowledge of a dynamic core, this naive procedure will not suit.

\begin{lemma}
  \label{th:intersection}
  Let $\ell$ be an integer,
  let $\mathcal{H}_1,\ldots,\mathcal{H}_\ell$ be graph classes and let,
  for each $i \in \intv{1,\ell}$, 
  $\acore_i$ be a dynamic core that solves $\mathcal{H}_i$.
  There exists a dynamic core $\acore$ that solves $\mathcal{H} = \bigcap_{i \in \intv{1,\ell}} \mathcal{H}_i$ such that, for each graph $G$, each rooted tree decomposition $\mathcal{D} = (T,r,\mathcal{X})$ of $G$, and each $t \in V(T)$:
  \begin{itemize}
  \item $|\process_{\acore,\agraph,\mathcal{D}}(t)| = \prod_{i \in \intv{1,\ell}} |\process_{\acore_i,G,\mathcal{D}}(t)|$
  \item $\timecomplexity_g(\acore,G,\mathcal{D}) = \sum_{i \in \intv{1,\ell}}\timecomplexity_g(\acore_i,G,\mathcal{D}) + \Ocal\left(\sum_{t \in V(T)} |\process_{\acore,G,\mathcal{D}}(t)|\right)$
  \end{itemize}
\end{lemma}
\begin{proof}
  We define $\acore$ such that for each $G,G',G'' \in \graphs$,
  \begin{align*}
  \accepting_{\acore}(G) &= \{(m_1,\ldots, m_\ell) \mid \forall i \in \intv{1,\ell}, m_i \in \accepting_{\acore_i}(G)\ \},\\
  \process_{\acore,0}(G) &= \{(m_1,\ldots, m_\ell) \mid \forall i \in \intv{1,\ell},  m_i \in \process_{\acore_i,0}(G)\ \},\\
    \process_{\acore,1}(G,G') &= \{((m_1,\ldots, m_\ell),(m'_1,\ldots, m'_\ell)) \mid\\
         & ~~~~~~~~~~~~~~~~~~~~~~~~~~ \forall i  \in \intv{1,\ell}, (m_i,m'_i) \in \process_{\acore_i,1}(G,G')\ \},\\
           \process_{\acore,2}(G,G',G'') &= \{((m_1,\ldots, m_\ell),(m'_1,\ldots, m'_\ell),(m''_1,\ldots, m''_\ell)) \mid \\
                         & ~~~~~~~~~~~~~~~~~~~~~~~~~~\forall i \in \intv{1,\ell}, (m_i,m'_i,m''_i) \in \process_{\acore_i,2}(G,G',G'')\ \}.\\
  \end{align*}

  We now prove that $\acore$ solves $\mathcal{H}$.
  First note that $\Sigma_{\acore} = \Gamma \cup \bigcup_{i \in \intv{1,\ell}}\Sigma_{\acore_i}$.
  Let $G$ be a graph and $\mathcal{D}  = (T,r,\mathcal{X})$ be a rooted tree decomposition of $G$.

  Assume first that $G \in \mathcal{H}$, then by definition of $\mathcal{H}$, for each $i \in \intv{1,\ell}$,  $G \in \mathcal{H}_i$.
  Let, for each $i \in \intv{1,\ell}$, $\alpha_i: V(T) \to \word_{\acore_i}$ be a $(\acore_i,G,\mathcal{D})$-witness.
  Note that it exists as $\acore_i$ solves $\mathcal{H}_i$ and $G \in \mathcal{H}_i$.
  We define $\alpha: V(T) \to \word_\acore$ such that for each $t \in V(T)$, $\alpha(t) = (\alpha_1(t), \ldots, \alpha_\ell(t))$.
  By construction of $\acore$, $\alpha$ is a $(\acore,G,\mathcal{D})$-witness.

  Assume now that there exists a  $(\acore,G,\mathcal{D})$-witness $\alpha: V(T) \to \word_\acore$.
  Then by definition of $\acore$, for each $t \in V(T)$, $\alpha(t)$ is a $\ell$-tuple.
  Let define $\ell$ functions $\alpha_i:V(T) \to \word_\acore$, $i \in \intv{1,\ell}$, such that for each $i \in \intv{1,\ell}$ and each $t \in V(T)$, $\alpha_i(t) = \alpha(t).(i)$.
  Then, by definition of $\acore$, for each $i \in \intv{1,\ell}$, $\alpha_i$ is a $(\acore_i,G,\mathcal{D})$-witness and so $G \in \mathcal{H}_i$.
  Thus $G \in \mathcal{H}$.

  Let us now analyze the needed running time for this algorithm.
  Let $G$ be a graph and $\mathcal{D} = (T,r,\mathcal{X})$ be a rooted tree decomposition of $G$.
  For each $t \in V(T)$, we have, by definition, $|\process_{\acore,G,\mathcal{D}}(t)| = \prod_{i \in \intv{1,\ell}} |\process_{\acore_i,G,\mathcal{D}}(t)|$.
  In order to construct the data of $\acore$ associated to $(G,\mathcal{D})$, we need first to construct the data of $\acore_i$ associated to $(G,\mathcal{D})$ for each $i \in \intv{1,\ell}$ and then to combine them.
  Thus, we have
  \begin{displaymath}
    \timecomplexity_g(\acore,G,\mathcal{D}) =\sum_{i \in \intv{1,\ell}}\timecomplexity_g(\acore_i,G,\mathcal{D}) +  \Ocal\left(\sum_{t \in V(T)} |\process_{\acore,G,\mathcal{D}}(t)|\right).
\end{displaymath}
\end{proof}

\begin{lemma}
  \label{th:union}
  Let $\ell$ be an integer,
  let $\mathcal{H}_1,\ldots,\mathcal{H}_\ell$ be graph classes and let,
  for each $i \in \intv{1,\ell}$, 
  $\acore_i$ be a dynamic core that solves $\mathcal{H}_i$.
  There exists a dynamic core $\acore$ that solves $\mathcal{H} = \bigcup_{i \in \intv{1,\ell}} \mathcal{H}_i$ such that, for each graph $G$, each rooted tree decomposition $\mathcal{D}= (T,r,\mathcal{X})$ of $G$, and each $t \in V(T)$:
  \begin{itemize}
  \item $|\process_{\acore,\agraph,\mathcal{D}}(t)| = \sum_{i \in \intv{1,\ell}} |\process_{\acore_i,G,\mathcal{D}}(t)|$
  \item $\timecomplexity_g(\acore,G,\mathcal{D}) = \sum_{i \in \intv{1,\ell}}\timecomplexity_g(\acore_i,G,\mathcal{D}) + \Ocal\left(\sum_{t \in V(T)} |\process_{\acore,G,\mathcal{D}}(t)|\right)$
  \end{itemize}
\end{lemma}
\begin{proof}
  Let $\bot$ be an unused letter. We define $\acore$ such that for each $G,G',G'' \in \graphs$,
  \begin{align*}
  \accepting_{\acore}(G) &= \{(m_1,\ldots, m_\ell) \mid \exists i \in \intv{1,\ell}, m_i \in \accepting_{\acore_i}(G)\ \},\\
    \process_{\acore,0}(G) &= \{(m_1,\ldots, m_\ell) \mid\\
                           & ~~~~~~~~~~~~~~~~~~ \exists i \in \intv{1,\ell},  m_i \in \process_{\acore_i,0}(G)\\
                         & ~~~~~~~~~~~~~~~~~~ \forall j \in \intv{1,\ell} \setminus \{i\}, m_j = \bot \ \},\\
    \process_{\acore,1}(G,G') &= \{((m_1,\ldots, m_\ell),(m'_1,\ldots, m'_\ell)) \mid\\
     & ~~~~~~~~~~~~~~~~~~\exists i \in \intv{1,\ell}, (m_i,m'_i) \in \process_{\acore_i,1}(G,G')\\
                         & ~~~~~~~~~~~~~~~~~~ \forall j \in \intv{1,\ell} \setminus \{i\}, (m_j,m'_j) = (\bot,\bot) \ \}, and\\
    \process_{\acore,2}(G,G',G'') &= \{((m_1,\ldots, m_\ell),(m'_1,\ldots, m'_\ell),(m''_1,\ldots, m''_\ell)) \mid \\
     & ~~~~~~~~~~~~~~~~~~\forall i \in \intv{1,\ell}, (m_i,m'_i,m''_i) \in \process_{\acore_i,2}(G,G',G'')\\
                         & ~~~~~~~~~~~~~~~~~~ \forall j \in \intv{1,\ell} \setminus \{i\}, (m_j,m'_j,m''_j) = (\bot,\bot,\bot) \ \}.\\
  \end{align*}
  We prove, using the same kind of argumentation as for Theorem~\ref{th:intersection}, that $\acore$ solves $\mathcal{H}$.
  Note that in this case, the letter $\bot$ is used, in particular, for each coordinate $i$ such that $G \not \in \mathcal{H}_i$.
\end{proof}

\section{Main theorem}
\label{sec:main}
In this section we show,
given two integers $\ell$ and $p$, $\ell$ graph classes $\mathcal{H}_1,\ldots,\mathcal{H}_\ell$ and $\ell$ dynamic cores $\acore_1,\ldots,\acore_\ell$ such that, for each $i \in \intv{1,\ell}$, $\acore_i$ solves $\mathcal{H}_i$,
how to construct dynamic cores that solve
$\graphpartition(\mathcal{H}_1, \ldots, \mathcal{H}_\ell,p)$,
$\partition(\mathcal{H}_1, \ldots, \mathcal{H}_\ell)$, and
$\edgepartition(\mathcal{H}_1, \ldots, \mathcal{H}_\ell)$.

\bigskip

We start by the graph partition problem, that is the most involved case.

\begin{theorem}
  \label{th:graphpartition}
  Let $\ell$ and $p$ be two integers,
  let $\mathcal{H}_1,\ldots,\mathcal{H}_\ell$ be graph classes and let,
  for each $i \in \intv{1,\ell}$, 
  $\acore_i$ be a dynamic core that solves $\mathcal{H}_i$.
  There exists a dynamic core $\acore$ that solves $\mathcal{H} = \graphpartition(\mathcal{H}_1, \ldots, \mathcal{H}_\ell,p)$ such that, for each graph $G$, each rooted tree decomposition $\mathcal{D}= (T,r,\mathcal{X})$ of $G$, and each $t \in V(T)$:
  \begin{itemize}
  \item $|\process_{\acore,G,\mathcal{D}}(t)| \leq \ell^{|Y_t|} \cdot (p+1)^2\cdot \prod_{i \in \intv{1,\ell}}\mathsf{size}(\acore_i,\agraph,\mathcal{D},t)$
  \item $\timecomplexity_g(\acore,G,\mathcal{D}) =  \sum_{t \in V(T)}\left( \mathcal{O}(|\process_{\acore,G,\mathcal{D}}(t)|) + 2^{|Y_t|} \cdot \sum_{i \in \intv{1,\ell}} \timecomplexity_n(\acore_i,G, \mathcal{D},t)\right)$

  \end{itemize}
\end{theorem}
\begin{proof}
  We define $\acore$ such that for each $G,G',G'' \in \graphs$,
  \begin{align*}
    \accepting_{\acore}(G) &= \{((m_1,V_1),\ldots, (m_\ell,V_\ell),q) \mid \\
                           &~~~~~~~~~~~~q \leq p,\\
                           &~~~~~~~~~~~~V_1,\ldots,V_\ell \text{ is a partition of $V(G)$, and }\\
                           &~~~~~~~~~~~~\forall i \in \intv{1,\ell}, m_i \in \accepting_{\acore_i}(G[V_i])\ \},\\
    \process_{\acore,0}(G) &= \{((m_1,V_1),\ldots, (m_\ell,V_\ell),0) \mid \\
                           &~~~~~~~~~~~~V_1,\ldots,V_\ell \text{ is a partition of $V(G)$ and }\\
                           &~~~~~~~~~~~~\forall i \in \intv{1,\ell},  m_i \in \process_{\acore_i,0}(G[V_i])\ \},\\
    \process_{\acore,1}(G,G') &= \{(((m_1,V_1),\ldots, (m_\ell,V_\ell),q),((m'_1,V'_1),\ldots, (m'_\ell,V'_\ell),q')) \mid\\
                           &~~~~~~~~~~~~U_1,\ldots,U_\ell \text{ is a partition of $V(G) \cup V(G')$, }\\
                           &~~~~~~~~~~~~\forall i  \in \intv{1,\ell}, V_i = U_i \cap V(G),\\
                           &~~~~~~~~~~~~\forall i  \in \intv{1,\ell}, V'_i = U_i \cap V(G'),\\
                           &~~~~~~~~~~~~\forall i  \in \intv{1,\ell}, (m_i,m'_i) \in \process_{\acore_i,1}(G[V_i],G'[V'_i]),\\
                           &~~~~~~~~~~~~q \leq p,\text{ and}\\
                           &~~~~~~~~~~~~q = q' + \sum_{i \in \intv{1,\ell}} |\mathsf{edge}_G(U_i \setminus V(G), (V(G) \cup V(G'))\setminus U_i)|\ \},\\
  \end{align*}
  \begin{align*}
    \process_{\acore,2}(G,G',G'') &= \{(((m_1,V_1),\ldots, (m_\ell,V_\ell),q),((m'_1,V'_1),\ldots, (m'_\ell,V'_\ell),q'),\\ &~~~~~~~~~~~~~~~~~~~~~~~~~~~~~~~~~~~~~~~~~~~~~~~((m''_1,V''_1),\ldots, (m''_\ell,V''_\ell),q'')) \mid \\
                           &~~~~~~~~~~~~U_1,\ldots,U_\ell \text{ is a partition of $V(G) \cup V(G') \cup V(G'')$, }\\
                           &~~~~~~~~~~~~\forall i  \in \intv{1,\ell}, V_i = U_i \cap V(G),\\
                           &~~~~~~~~~~~~\forall i  \in \intv{1,\ell}, V'_i = U_i \cap V(G'),\\
                           &~~~~~~~~~~~~\forall i  \in \intv{1,\ell}, V''_i = U_i \cap V(G''),\\
                           &~~~~~~~~~~~~\forall i  \in \intv{1,\ell}, (m_i,m'_i,m''_i) \in \process_{\acore_i,2}(G[V_i],G'[V'_i],G''[V''_i]),\\
                           &~~~~~~~~~~~~q \leq p,\text{ and}\\
                           &~~~~~~~~~~~~q = q' + q'' +\sum_{i \in \intv{1,\ell}} |\mathsf{edge}_G(U_i \setminus V(G), (V(G) \cup V(G')\cup V(G''))\setminus U_i)|\ \}.\\
  \end{align*}

  We now prove that $\acore$ solves $\mathcal{H} = \graphpartition(\mathcal{H}_1, \ldots, \mathcal{H}_\ell,p)$.
  First note that $\Sigma_{\acore} = \Gamma \cup \mathbb{N} \cup \bigcup_{i \in \intv{1,\ell}}\Sigma_{\acore_i}$.
  Let $G$ be a graph and $\mathcal{D}  = (T,r,\mathcal{X})$ be a rooted tree decomposition of $G$.

  Assume first that $G \in \mathcal{H}$.
  Then, by definition, there exists $V_1, \ldots, V_\ell$, partition of $V(G)$ such that for each $i \in \intv{1,\ell}$, $G[V_i] \in \mathcal{H}_i$ and $\sum_{1 \leq i < j \leq \ell}|\mathsf{edge}_G(V_i,V_j)| \leq p$.
  As, for each $i \in \intv{1,\ell}$, $\mathcal{D}[V_i]$ is a rooted tree decomposition of $G[V_i]$ and $G[V_i] \in \mathcal{H}_i$, there exists  $\alpha_i: V(T) \to \word_{\acore_i}$, a $(\acore_i,G[V_i],\mathcal{D}[V_i])$-witness.
  Note that this witness exists as $G[V_i] \in \mathcal{H}_i$.
  We define $\alpha: V(T) \to \word_\acore$. 
  such that for each
  $t \in V(T)$,

  \begin{displaymath}
  \alpha(t) = ((\alpha_1(t),X_t \cap V_1), \ldots, (\alpha_\ell(t),X_t \cap V_\ell),
  \sum_{i \in \intv{1,\ell}} |\mathsf{edge}_G((Z_t \cap V_i) \setminus X_t, Z_t \setminus V_i)|).
\end{displaymath}
  By construction of $\acore$, $\alpha$ is a $(\acore,G,\mathcal{D})$-witness.

  Assume now that there exists a  $(\acore,G,\mathcal{D})$-witness $\alpha: V(T) \to \word_\acore$.
  Then by definition of $\acore$, for each $t \in V(T)$, $\alpha(t)$ is a $(\ell+1)$-tuple where the $\ell$ first coordinates are pairs with the shape $(m,V)$ where $m \in \word_\acore$ and $V \subseteq V(G)$, and where $\alpha(t).(\ell+1)$ is an integer in $\intv{0,p}$.
  For each $i \in \intv{1,\ell}$, let $V_i = \bigcup_{t \in V(T)} \alpha(t).(i).(2)$, and let $\alpha_i: V(T) \to \word_\acore$ be such that for each $t \in V(T)$, $\alpha_i(t) = \alpha(t).(i).(1)$.
  Then by definition of $\acore$, for each $i \in \intv{1,\ell}$, $\alpha_i$ is a $(\acore_i,G[V_i],\mathcal{D}_i)$-witness, and so, $G[V_i] \in \mathcal{H}_i$.
  Note also that by definition of $\accepting_{\acore}$, $\process_{\acore,0}$, $\process_{\acore,1}$, and $\process_{\acore,2}$, the partition selected by $\alpha$ at step $t \in V(T)\setminus \{r\}$ is consistent with the one selected at step $t'$ where $t'$ is the parent of $t$. Combined with the definition of a tree decomposition, we obtain that $(V_1, \ldots, V_\ell)$ is a partition of $V(G)$.
  Moreover, as $\alpha$ is a $(\acore,G,\mathcal{D})$-witness, we have that $\alpha(r).(\ell+1) \leq p$, and so $G \in \mathcal{H}$.

  Let us now analyze the needed running time for this algorithm.
  Let $G$ be a graph and $\mathcal{D} = (T,r,\mathcal{X})$ be a rooted tree decomposition of $G$.
  Then for each partition
  $V_1, \ldots, V_\ell$ of $Y_t$,
  there is at most $(p+1)^2$ ways to combine $p'$ and $p''$ and so
  there are at most $(p+1)^2\cdot \prod_{i \in \intv{1,\ell}}\size(\acore_i,\agraph,\mathcal{D},t)$ ways to construct an element of $\process_{\acore,G,\mathcal{D}}(t)$ consistent with the partition.
  Moreover, we have $\ell^{|Y_t|}$ possible partitions of the set $|Y_t|$ into $\ell$ sets.
  Thus $|\process_{\acore,G,\mathcal{D}}(t)| \leq \ell^{|Y_t|} \cdot (p+1)^2\cdot \prod_{i \in \intv{1,\ell}}\mathsf{size}(\acore_i,\agraph,\mathcal{D},t)$.
  In order to construct the data of $\acore$ associated to $(G,\mathcal{D})$, we first need, for each $i \in \intv{1,\ell}$, to construct the data of $\acore_i$ associated to $(G[V],\mathcal{D}[V],t)$ for each $t \in V(T)$ and $V \subseteq Y_t$,  and then, for each $t \in V(T)$ try every partition of $Y_t$  and combine the corresponding data accordingly.
  Thus, we have
  \begin{align*}
    \forall t \in V(T),\ \timecomplexity_n(\acore,G,\mathcal{D},t) &= \mathcal{O}( |\process_{\acore,G,\mathcal{D}}(t)| ) +  2^{|Y_t|} \cdot \sum_{i \in \intv{1,\ell}} \timecomplexity_n(\acore_i,G, \mathcal{D},t) \text{ and}\\
    \timecomplexity_g(\acore,G,\mathcal{D}) &=  \sum_{t \in V(T)}\timecomplexity_n(\acore,G,\mathcal{D},t).
  \end{align*}
  The theorem follows.
  \end{proof}

  Coloring problems are graph partition problems where it is not needed to keep track of the number of transversal edges.
  Thus the dynamic cores we present for the coloring problems are simpler than the one providing for the graph partition problems.

\begin{theorem}
  \label{th:partition}
  Let $\ell$ be an integer,
  let $\mathcal{H}_1,\ldots,\mathcal{H}_\ell$ be graph classes and let,
  for each $i \in \intv{1,\ell}$, 
  $\acore_i$ be a dynamic core that solves $\mathcal{H}_i$.
  There exists a dynamic core $\acore$ that solves $\mathcal{H} = \partition(\mathcal{H}_1, \ldots, \mathcal{H}_\ell)$ such that, for each graph $G$, each rooted tree decomposition $\mathcal{D}= (T,r,\mathcal{X})$ of $G$, and each $t \in V(T)$:
  \begin{itemize}
  \item $|\process_{\acore,G,\mathcal{D}}(t)| \leq \ell^{|Y_t|} \cdot \prod_{i \in \intv{1,\ell}}\mathsf{size}(\acore_i,\agraph,\mathcal{D},t)$
  \item $\timecomplexity_g(\acore,G,\mathcal{D}) =  \sum_{t \in V(T)}\left( \mathcal{O}(|\process_{\acore,G,\mathcal{D}}(t)|) + 2^{|Y_t|} \cdot \sum_{i \in \intv{1,\ell}} \timecomplexity_n(\acore_i,G, \mathcal{D},t)\right)$

  \end{itemize}
\end{theorem}
\begin{proof}
  Using the same base as for Theorem~\ref{th:graphpartition}, we define $\acore$ such that for each $G,G',G'' \in \graphs$,
  \begin{align*}
    \accepting_{\acore}(G) &= \{((m_1,V_1),\ldots, (m_\ell,V_\ell)) \mid \\
                           &~~~~~~~~~~~~V_1,\ldots,V_\ell \text{ is a partition of $V(G)$ and }\\
                           &~~~~~~~~~~~~\forall i \in \intv{1,\ell}, m_i \in \accepting_{\acore_i}(G[V_i])\ \},\\
    \process_{\acore,0}(G) &= \{((m_1,V_1),\ldots, (m_\ell,V_\ell)) \mid \\
                           &~~~~~~~~~~~~V_1,\ldots,V_\ell \text{ is a partition of $V(G)$ and }\\
                           &~~~~~~~~~~~~\forall i \in \intv{1,\ell},  m_i \in \process_{\acore_i,0}(G[V_i])\ \},\\
    \process_{\acore,1}(G,G') &= \{(((m_1,V_1),\ldots, (m_\ell,V_\ell)),((m'_1,V'_1),\ldots, (m'_\ell,V'_\ell))) \mid\\
                           &~~~~~~~~~~~~U_1,\ldots,U_\ell \text{ is a partition of $V(G) \cup V(G')$, }\\
                           &~~~~~~~~~~~~\forall i  \in \intv{1,\ell}, V_i = U_i \cap V(G),\\
                           &~~~~~~~~~~~~\forall i  \in \intv{1,\ell}, V'_i = U_i \cap V(G'),\text{ and}\\
                           &~~~~~~~~~~~~\forall i  \in \intv{1,\ell}, (m_i,m'_i) \in \process_{\acore_i,1}(G[V_i],G'[V'_i])\ \},\\
    \process_{\acore,2}(G,G',G'') &= \{(((m_1,V_1),\ldots, (m_\ell,V_\ell)),((m'_1,V'_1),\ldots, (m'_\ell,V'_\ell)),((m''_1,V''_1),\ldots, (m''_\ell,V''_\ell))) \mid \\
                           &~~~~~~~~~~~~U_1,\ldots,U_\ell \text{ is a partition of $V(G) \cup V(G') \cup V(G'')$, }\\
                           &~~~~~~~~~~~~\forall i  \in \intv{1,\ell}, V_i = U_i \cap V(G),\\
                           &~~~~~~~~~~~~\forall i  \in \intv{1,\ell}, V'_i = U_i \cap V(G'),\\
                           &~~~~~~~~~~~~\forall i  \in \intv{1,\ell}, V''_i = U_i \cap V(G''),\text{ and}\\
                           &~~~~~~~~~~~~\forall i  \in \intv{1,\ell}, (m_i,m'_i,m''_i) \in \process_{\acore_i,2}(G[V_i],G'[V'_i],G''[V''_i])\ \}.\\
  \end{align*}

  The proof that $\acore$ solves $\mathcal{H} = \partition(\mathcal{H}_1, \ldots, \mathcal{H}_\ell)$ is omitted as it is  similar to the one provided for Theorem~\ref{th:graphpartition} at the difference that this time we do not keep track of the transversal edges. 
  \end{proof}

  Edge partitioning problems are really similar to coloring problem at the difference that the subgraphs we consider are induced by a set of edges instead of a set of vertices. 

  \begin{theorem}
  \label{th:edgepartition}
  Let $\ell$ be an integer,
  let $\mathcal{H}_1,\ldots,\mathcal{H}_\ell$ be graph classes and let,
  for each $i \in \intv{1,\ell}$, 
  $\acore_i$ be a dynamic core that solves $\mathcal{H}_i$.
  There exists a dynamic core $\acore$ that solves $\mathcal{H} = \edgepartition(\mathcal{H}_1, \ldots, \mathcal{H}_\ell)$ such that, for each graph $G$, each rooted tree decomposition $\mathcal{D}= (T,r,\mathcal{X})$ of $G$, and each $t \in V(T)$:
  \begin{itemize}
  \item$|\process_{\acore,G,\mathcal{D}}(t)| \leq \ell^{|E(G[Y_t])|} \cdot \prod_{i \in \intv{1,\ell}}\mathsf{size}(\acore_i,\agraph,\mathcal{D},t)$
  \item
        $    \timecomplexity_g(\acore,G,\mathcal{D}) =  \sum_{t \in V(T)}\left( \mathcal{O}(|\process_{\acore,G,\mathcal{D}}(t)|) + 2^{|E(G[Y_t])|} \cdot \sum_{i \in \intv{1,\ell}} \timecomplexity_n(\acore_i,G, \mathcal{D},t)\right)$

  \end{itemize}
\end{theorem}
\begin{proof}
  We define $\acore$ such that for each $G,G',G'' \in \graphs$,
  \begin{align*}
    \accepting_{\acore}(G) &= \{((m_1,E_1),\ldots, (m_\ell,E_\ell)) \mid \\
                           &~~~~~~~~~~~~E_1,\ldots,E_\ell \text{ is a partition of $E(G)$ and }\\
                           &~~~~~~~~~~~~\forall i \in \intv{1,\ell}, m_i \in \accepting_{\acore_i}(G[E_i])\ \},\\
    \process_{\acore,0}(G) &= \{((m_1,E_1),\ldots, (m_\ell,E_\ell)) \mid \\
                           &~~~~~~~~~~~~E_1,\ldots,E_\ell \text{ is a partition of $E(G)$ and }\\
                           &~~~~~~~~~~~~\forall i \in \intv{1,\ell},  m_i \in \process_{\acore_i,0}(G[E_i])\ \},\\
    \process_{\acore,1}(G,G') &= \{(((m_1,E_1),\ldots, (m_\ell,E_\ell)),((m'_1,E'_1),\ldots, (m'_\ell,E'_\ell))) \mid\\
                           &~~~~~~~~~~~~U_1,\ldots,U_\ell \text{ is a partition of $E(G) \cup E(G')$, }\\
                           &~~~~~~~~~~~~\forall i  \in \intv{1,\ell}, E_i = U_i \cap E(G),\\
                           &~~~~~~~~~~~~\forall i  \in \intv{1,\ell}, E'_i = U_i \cap E(G'),\text{ and}\\
                           &~~~~~~~~~~~~\forall i  \in \intv{1,\ell}, (m_i,m'_i) \in \process_{\acore_i,1}(G[E_i],G'[E'_i])\ \},\\
    \process_{\acore,2}(G,G',G'') &= \{(((m_1,E_1),\ldots, (m_\ell,E_\ell)),((m'_1,E'_1),\ldots, (m'_\ell,E'_\ell)),((m''_1,E''_1),\ldots, (m''_\ell,E''_\ell))) \mid \\
                           &~~~~~~~~~~~~U_1,\ldots,U_\ell \text{ is a partition of $E(G) \cup E(G') \cup E(G'')$, }\\
                           &~~~~~~~~~~~~\forall i  \in \intv{1,\ell}, E_i = U_i \cap E(G),\\
                           &~~~~~~~~~~~~\forall i  \in \intv{1,\ell}, E'_i = U_i \cap E(G'),\\
                           &~~~~~~~~~~~~\forall i  \in \intv{1,\ell}, E''_i = U_i \cap E(G''),\text{ and}\\
                           &~~~~~~~~~~~~\forall i  \in \intv{1,\ell}, (m_i,m'_i,m''_i) \in \process_{\acore_i,2}(G[E_i],G'[E'_i],G''[E''_i])\ \}.\\
  \end{align*}
  The proof that $\acore$ solves $\mathcal{H} = \edgepartition(\mathcal{H}_1, \ldots, \mathcal{H}_\ell)$ is omitted as it is, again, similar to the one provided for Theorem~\ref{th:graphpartition} at the difference that this time we partition over the edges instead of the vertices and there are no transversal edges to consider.
\end{proof}
\section{Applications}
\label{sec:application}
In this section we show how our results lead to significant simplification when looking for algorithms parameterized by treewidth.
We first confront our approach against well-known problems, namely \textsc{Vertex Cover} and \textsc{Graph $q$-Coloring}, showing that we obtain comparable running time.
We then show how this leads to quickly obtain algorithms for exotic problems, $\partition(\mathcal{T},\mathcal{T})$ in our example, for which describing an algorithm in the usual way would have been long and tedious.

\textsc{Vertex Cover}, corresponding to $\partition(\graphs_k,\mathcal{I})$ for some integer $k$, is well known to be solvable in time $2^{\w}\cdot n^{\mathcal{O}(1)}$ when a tree decomposition of width $\w$ of the input graph is given, see for instance~\cite[Theorem 7.9]{CyFoKoLoMaPiPiSa2015}, while, as shown by \cite{ImPaZa2001}, no algorithm running in time $2^{o(\tw)} \cdot n^{\mathcal{O}(1)}$ can solve it unless ETH fails.
Combining Observation~\ref{obs:indep}, Observation~\ref{obs:boudedgraphs}, and Theorem~\ref{th:partition}, we obtain a dynamic core that solves \textsc{Vertex Cover}.
Moreover, combined with Theorem~\ref{theorem:DynamicSolvability}, we obtain an algorithm solving \textsc{Vertex Cover} in time $2^{\w}\cdot n^{\mathcal{O}(1)}$ when a tree decomposition of width $\w$ of the input graph is given.

More generally, deletion problems are problems that attract a lot of attention and that can be considered as coloring problems.
Indeed, given a graph class $\mathcal{H}$, 
the \textsc{$\mathcal{H}$-deletion} corresponds to the problem $\partition(\graphs_k,\mathcal{H})$, for some integer $k$.
Moreover, we show, in Observation~\ref{obs:boudedgraphs} that \textsc{$\graphs_p$-recognition}, for some integer $p$, has a polynomial dynamic core.
Combining Observation~\ref{obs:boudedgraphs} with Theorem~\ref{th:partition}, we obtain that
if there exists a dynamic core such that recognizing \textsc{$\mathcal{H}$} can be done in time $2^{f(\tw)}\cdot n^{\mathcal{O}(1)}$ for some function $f$, then \textsc{$\mathcal{H}$-deletion} can be solved in time $2^{\mathcal{O}(\tw+f(\tw))}\cdot n^{\mathcal{O}(1)}$.

\bigskip

The most basic and well-known coloring problem is \textsc{Graph $q$-Coloring}. Again we obtain an asymptotically optimal algorithm, see for instance~\cite[Theorem 14.6 and Theorem 14.41]{CyFoKoLoMaPiPiSa2015}, by combining Observation~\ref{obs:indep} and Theorem~\ref{th:partition}.

\begin{corollary}
  \textsc{Graph $q$-Coloring} can be solved in time $q^{\w}\cdot n^{\mathcal{O}(1)}$ when a tree decomposition of width $\w$ of the input graph is given.
\end{corollary}
\begin{proof}
  Given a fixed integer $q$,
  the \textsc{Graph $q$-Coloring} problem corresponds to $\partition(\mathcal{H}_1, \ldots, \mathcal{H}_q)$ where for each $i \in \intv{1,q}$, $\mathcal{H}_i = \mathcal{I}$.
  The result follows from the combination of Observation~\ref{obs:indep} and Theorem~\ref{th:partition}.
\end{proof}

\bigskip

As discussed in the introduction, finding an algorithm for $\partition(\mathcal{T},\mathcal{T})$ parameterized by the treewidth of the input graph is a question of interest in bioinformatics.
By combining  Observation~\ref{obs:tree}, Theorem~\ref{th:partition}, and Theorem~\ref{theorem:DynamicSolvability}, we obtain an efficient algorithm solving $\partition(\mathcal{T},\mathcal{T})$.

\begin{corollary}
  \label{coro:tt}
  There exists an algorithm that solves $\partition(\mathcal{T},\mathcal{T})$ in time $2^{\mathcal{O}(\tw)}\cdot n^{\mathcal{O}(1)}$.
\end{corollary}

To the best of our knowledge, the only other algorithm parameterized by treewidth for $\partition(\mathcal{T},\mathcal{T})$ was known using Courcelle's theorem and run in time $2^{2^{\Ocal(\tw)}}\cdot n^{\mathcal{O}(1)}$.

\section{Conclusion}
\label{sec:conclusion}

In this paper, we provide a generic tool for solving graph partition problems, coloring problems, and edge partition problems parameterized by the treewidth of the input graph.
In particular, the developed approach provides a way, different than Courcelle's theorem, to determine whether a problem is FPT parameterized by treewidth
and  also provides a first estimation of the expected running time of an algorithm solving the given problem.
We would like to highlight the quality of these estimations as, for some well-known problems, they are asymptotically optimal.

In this conclusion we want to stress that when solving graph partition problems, we count the number of transversal edges.
The illustrated technique allows, for instance, with some small modifications, to count separately transversal edges between different vertex sets of the partition.
One can ask for instance for a partition of the vertex set of the input graph into three sets $V_1$, $V_2$, $V_3$ such that there are $k_{1,2}$ edges between $V_1$ and $V_2$, at most $k_{2,3}$ edges edges between $V_2$ and $V_3$, and no edge between $V_1$ and $V_3$ for some integers $k_{1,2}$ and $k_{2,3}$.

This approach also allows us to work with balanced partition.
In this case, we first need to obtain the size of the input graph before constructing the dynamic core.
Indeed we will need to intersect each graph class we consider with the class of graph of size $q$ (or $q+1$ if needed) for some correctly calculated $q$.
For normalized cuts, this trick will not work as the target graph classes are not fixed before the algorithm starts to run.

More generally, we believe that using the dynamic programming core model, one can easily compose dynamic programming tree-decomposition-based algorithms with several added constraints.
Moreover we also believe that it is adapted to quickly study exotic problems parameterized by treewidth.

\bibliographystyle{abbrvnat}
% use the following instead if you encounter problems 
%\bibliographystyle{alpha}
\bibliography{biblio}
\label{sec:biblio}

\end{document}